\documentclass[12pt]{article}

\usepackage{amsmath}
\usepackage{amssymb}
\usepackage{amsfonts}
\usepackage{latexsym}
\usepackage{color}

\catcode `\@=11 \@addtoreset{equation}{section}
\def\theequation{\arabic{section}.\arabic{equation}}
\catcode `\@=12

\newcommand{\be}{\begin{equation}}
	\newcommand{\en}{\end{equation}}
\newcommand{\bea}{\begin{eqnarray}}
	\newcommand{\ena}{\end{eqnarray}}
\newcommand{\beano}{\begin{eqnarray*}}
	\newcommand{\enano}{\end{eqnarray*}}
\newcommand{\bee}{\begin{enumerate}}
	\newcommand{\ene}{\end{enumerate}}

\newcommand{\mc}{\mathcal}

\newcommand{\D}{{\mc D}}
\newcommand{\V}{{\mc V}}

\newcommand{\Sc}{{\cal S}}
\newcommand{\E}{{\cal E}}
\newcommand{\F}{{\cal F}}
\newcommand{\G}{{\cal G}}

\newcommand{\Lc}{{\cal L}}
\newcommand{\ltwo}{{\Lc^2(\mathbb{R})}}
\newcommand{\scr}{{\Sc(\mathbb{R})}}

\newcommand{\1}{1 \!\! 1}

\newcommand{\Hil}{\mc H}

\newtheorem{thm}{Theorem}

\newtheorem{lemma}[thm]{Lemma}

\newtheorem{defn}[thm]{Definition}

\newenvironment{proof}{\noindent {\bf Proof --}}{\hfill$\square$ \vspace{3mm}\endtrivlist}

\catcode `\@=11 \@addtoreset{equation}{section}
\catcode `\@=12

\begin{document}

\thispagestyle{empty}

\vspace*{2cm}

\begin{center}
{\Large \bf A Swanson-like Hamiltonian and the inverted harmonic oscillator}   \vspace{2cm}\\

{\large Fabio Bagarello}\\
  Dipartimento di  Ingegneria, \\
Universit\`a di Palermo,\\ I-90128  Palermo, Italy,\\
e-mail: fabio.bagarello@unipa.it\\

\end{center}

\vspace*{1cm}

\begin{abstract}

We deduce the eigenvalues and the eigenvectors of a parameter-dependent Hamiltonian $H_\theta$ which is closely related to the Swanson Hamiltonian, and we construct bi-coherent states for it. After that, we show how and in which sense the eigensystem of the Hamiltonian $H$ of the inverted quantum harmonic oscillator can be deduced from that of $H_\theta$. We show that there is no need to introduce a different scalar product using some ad hoc metric operator, as suggested by other authors. Indeed we prove that a distributional approach is sufficient to deal with the Hamiltonian $H$ of the inverted oscillator.

\end{abstract}

\vspace{2cm}


\vfill


\newpage

\section{Introduction}

In the past years we have observed an always increasing interest for non-Hermitian Hamiltonians in quantum mechanics. This is because, even for quite non trivial systems like the cubic oscillator, it become evident that {\em strange} Hamiltonians may have real eigenvalues, even in presence of purely imaginary potentials, \cite{ben1}. Since then, a lot of work was done to achieve a deeper comprehension of this kind of operators, both for their physical relevance, and for their interesting mathematical properties. Some relevant monographs and edited books are  \cite{benbook}-\cite{bagspringer}.

Several Hamiltonians have been proposed by many authors along the years, some defined on finite-dimensional Hilbert spaces and some other in $\ltwo$, or in similar spaces. Quite often, an Hamiltonian $H$ which is not (Dirac-) self-adjoint, $H\neq H^\dagger$, but which has all its eigenvalues real, turns out to be self-adjoint with respect to a different conjugation map, defined by means of a different scalar product which is not, for instance, the natural one in $\ltwo$. In this case, some metric operator can be introduced in the game. This is essentially an invertible operator $S$ which is the key ingredient to define a new scalar product in the Hilbert space $\Hil$ starting from its natural one, $\langle.,.\rangle$: $\langle f,g\rangle_S=\langle Sf,g\rangle$. Of course, $S$ cannot be completely arbitrary. First of all, it must be positive, in order to ensure that $\langle f,f\rangle_S\geq0$ for all $f\in D(S)$, the domain of $S$. Moreover, the simplest situation is when both $S$ and $S^{-1}$ are bounded. In this case we can always assume that $D(S)=D(S^{-1})=\Hil$. However, this is not always guaranteed  if $\dim(\Hil)=\infty$. Actually, most of the times, in concrete examples, the opposite is true. This is the case of the  cubic oscillator mentioned above, \cite{petr1}, but also of many other systems which have been considered by many authors. We refer to \cite{baginbagbook} for an analysis of some of these systems from this point of view. 

Two operators which have been considered in this context in many details by several authors are the Swanson and the inverted quantum oscillator Hamiltonians, see, e.g., \cite{swan}-\cite{reboiro} for the first and \cite{barton}-\cite{marcucci} for the latter. In our knowledge the link between these Hamiltonians is not considered in the literature. This is possibly because the Swanson Hamiltonian $h_\theta$, which depends on the parameter $\theta\in\left]-\frac{\pi}{4},\frac{\pi}{4}\right[$ becomes singular in the limit $\theta\rightarrow\pm\,\frac{\pi}{4}$. In this paper we propose a new approach to the inverted quantum harmonic oscillator (IQHO) based on a Swanson-like Hamiltonian $H_\theta$ which does not become singular for any value of $\theta$. The analysis of the eigenvalues and the eigenvectors of $H_\theta$ is based on the construction of pseudo-bosonic ladder operators. The eigenvectors turn out to be in $\ltwo$, but only for some range of values of $\theta$. For these values we also construct bi-coherent states for the annihilation operators. Then we show that, taking a suitable limit on $\theta$, the IQHO can be recovered, together with its eigenvalues, eigenvectors and bi-coherent states. However, in this case, we have to leave $\ltwo$ and go to a suitable space of continuous functionals which, we believe, is the natural habitat for these states, and for the physical system.

The paper is organized as follows:

In the next section we propose our slightly modified version of the Swanson Hamiltonian, $H_\theta$, which can be rewritten in terms of pseudo-bosonic ladder operators, and we compute its eigenvalues and eigenvectors, imposing the constraint that these latter should be square-integrable. The same construction is repeated for $H_\theta^\dagger$.  The properties of these eigenvectors are discussed, and the bi-coherent states for the system are defined and studied.

Then, in Section \ref{sect3} we introduce the Hamiltonian $H$ of the IQHO and we show that $\ltwo$ is not the proper space to work with. Indeed we show that a distributional approach is more convenient, since the eigenvectors of $H$ can be found as weak limits of those of $H_\theta$, and that they are tempered distributions. We also show that a second space of continuous functionals can be constructed, which also contains the eigenvectors of $H$ and their associated bi-coherent states.
It can be useful to stress that, with respect to other approaches proposed in the literature on the IQHO, we work here with the natural scalar product in $\ltwo$ and with the sesquilinear form which extends it. In other words, we do not introduce any unessential metric operator.

Section \ref{sect4} contains our conclusions, while some technical aspects are discussed in the Appendixes at the end of the paper. A third Appendix is devoted to the main definitions of $\D$-pseudo-bosons, useful for make the paper more readable.

  \section{The Swanson-like Hamiltonian}\label{sect2}
  
  Let us consider the following Hamiltonian in $\ltwo$:
  \be
  H_\theta=\frac{1}{2}\left(p^2+e^{2i\theta}\Omega^2x^2\right),
  \label{21}\en
for $\theta\in[-\pi,\pi]$, for the moment, and $\Omega>0$. Here, as usual, $[x,p]=i\1$, $x=x^\dagger$ and $p=p^\dagger$. It is clear that, if $\theta=\pm\frac{\pi}{2}$, $H_\theta$ becomes the Hamiltonian of the IQHO, $H_\pm=\frac{1}{2}\left(p^2-\Omega^2x^2\right)$.

Before starting our analysis of $H_\theta$, it is interesting to sketch briefly its relation with the Hamiltonian of the Swanson model, \cite{baginbagbook},
$$
h_\varphi=\frac{1}{2}\left(p^2+x^2\right)-\frac{i}{2}\,\tan(2\varphi)\left(p^2-x^2\right)=\frac{e^{-2i\varphi}}{2\cos(2\varphi)}\left(p^2+x^2e^{4i\varphi}\right).
$$
Apart from the obvious differences ($\Omega=1$ and $\varphi=2\theta$), what makes $h_\varphi$ not so relevant for us, in view of our interest for the IQHO, is the fact that when $\varphi\rightarrow\frac{\pi}{4}$, which would correspond to $p^2-x^2$, $\cos(2\varphi)\rightarrow0$, and $h_\varphi$ is no longer defined. For this reason, accordingly with  \cite{davies}, we prefer to exclude the $\varphi$-dependent quantity in the definition of the Hamiltonian, and work directly with $H_\theta$ in (\ref{21}).

Let us introduce the operators
\be
A_\theta=\frac{1}{\sqrt{2\Omega}}\left(e^{i\theta/2}\Omega\,x+i\,e^{-i\theta/2}p\right), \quad B_\theta=\frac{1}{\sqrt{2\Omega}}\left(e^{i\theta/2}\Omega\,x-i\,e^{-i\theta/2}p\right),
\label{22}\en
for all possible $\theta$. It is clear that $A_\theta$ and $B_\theta$ are densely defined in $\ltwo$, since in particular any test function $f(x)\in\Sc(\mathbb{R})$, the Schwartz space, belongs to the domains of both these operators: $\Sc(\mathbb{R})\subseteq D(A_\theta)$ and $\Sc(\mathbb{R})\subseteq D(B_\theta)$, for all $\theta$. It is also clear that $A_\theta^\dagger\neq B_\theta$. Indeed we can check that, for instance on $\Sc(\mathbb{R})$,
\be
A_\theta^\dagger=\frac{1}{\sqrt{2\Omega}}\left(e^{-i\theta/2}\Omega\,x-i\,e^{i\theta/2}p\right), \quad B_\theta^\dagger=\frac{1}{\sqrt{2\Omega}}\left(e^{-i\theta/2}\Omega\,x+i\,e^{i\theta/2}p\right).
\label{23}\en
The set $\scr$ is stable under the action of all these operators.
Formulas (\ref{23}) show that
\be
A_\theta^\dagger=B_{-\theta}, \qquad B_\theta^\dagger=A_{-\theta}.
\label{24}\en
Moreover, it is easy to see that these operators obey pseudo-bosonic commutation rules, \cite{baginbagbook}:
\be
[A_\theta,B_\theta]f(x)=f(x)
\label{25}\en
for all $f(x)\in\scr$, and for all values of $\theta$. This is in agreement with the fact that, if $\theta=0$, we go back to the ordinary bosonic operators $c=\frac{\Omega x+ip}{\sqrt{2\Omega}}$ and $c^\dagger=\frac{\Omega x-ip}{\sqrt{2\Omega}}$, $[c,c^\dagger]=\1$. Indeed we have
$$
A_0=B_0^\dagger=c, \qquad B_0=A_0^\dagger=c^\dagger.
$$
In terms of the operators in (\ref{22}) $H_\theta$ can be rewritten as
\be
H_\theta=\Omega e^{i\theta}\left(B_\theta A_\theta+\frac{1}{2}\1\right).
\label{26}\en
Then, because of (\ref{24}),
we have that
\be
H_\theta^\dagger=\Omega e^{-i\theta}\left(A_\theta^\dagger B_\theta^\dagger+\frac{1}{2}\1\right)=H_{-\theta},
\label{27}\en
on $\scr$. To deduce the eigensystems of $H_\theta$ and $H_\theta^\dagger$, following the usual strategy adopted for ladder operators and for pseudo-bosons in particular, see the Appendix \ref{appendixC} or \cite{baginbagbook}, we should look now for the ground state of the two annihilation operators $A_\theta$ and $B_\theta^\dagger$. But, since  $B_\theta^\dagger=A_{-\theta}$, it is sufficient to solve the differential equation $A_\theta\varphi_0^{(\theta)}(x)=0$, since the solution of $B_\theta^\dagger\psi_0^{(\theta)}(x)=0$ can be deduced as $\psi_0^{(\theta)}(x)=\varphi_0^{(-\theta)}(x)$. Hence, recalling that $p=-i\frac{d}{dx}$, we find:
\be
\varphi_0^{(\theta)}(x)=N^{(\theta)}e^{-\frac{1}{2}\,\Omega e^{i\theta}x^2}, \qquad \psi_0^{(\theta)}(x)=N^{(-\theta)}e^{-\frac{1}{2}\,\Omega e^{-i\theta}x^2},
\label{28}\en
where $N^{(\pm\theta)}$ are normalization constants which will be fixed later. 
All along this section we will work with functions in $\ltwo$. This requirement will be relaxed when dealing with the IQHO, since square integrability is lost, in that case. From (\ref{28}) we see that the vacua are in $\ltwo$ if $\Re(e^{\pm i\theta})=\cos(\theta)>0$. For this reason, from now on, we will require that $\theta\in I=\left]-\frac{\pi}{2},\frac{\pi}{2}\right[$. This constraint reminds very much the similar one for the Swanson model, where it was needed both for keeping square-integrability of the eigenstates of the Hamiltonian, and to work with a well defined Hamiltonian\footnote{Here this last problem is solved ab initio. But the problem with the square-integrability is still there, and requires $\theta\in I$.}.

With this in mind, and using again the usual pseudo-bosonic approach, we can construct two families of square-integrable functions, $\F_\varphi^{(\theta)}=\{\varphi_n^{(\theta)}(x), \,n=0,1,2,\ldots\}$ and $\F_\psi^{(\theta)}=\{\psi_n^{(\theta)}(x), \,n=0,1,2,\ldots\}$, where
\be
\varphi_n^{(\theta)}(x)=\frac{B_\theta^n}{\sqrt{n!}}\,\varphi_0^{(\theta)}(x)=\frac{N^{(\theta)}}{\sqrt{2^n\,n!}}\,H_n\left(e^{i\theta/2}\sqrt{\Omega}\,x\right)e^{-\frac{1}{2}\,\Omega e^{i\theta}x^2}
\label{29}\en
and
\be
\psi_n^{(\theta)}(x)=\frac{{A_\theta^\dagger}^n}{\sqrt{n!}}\,\psi_0^{(\theta)}(x)=\varphi_n^{(-\theta)}(x)=\frac{N^{(-\theta)}}{\sqrt{2^n\,n!}}\,H_n\left(e^{-i\theta/2}\sqrt{\Omega}\,x\right)e^{-\frac{1}{2}\,\Omega e^{-i\theta}x^2}.
\label{210}\en
Here $H_n(x)$ is the $n$-th Hermite polynomial. The proof of these formulas, and in particular of (\ref{29}), is a consequence of the definition $\varphi_n^{(\theta)}(x)=\frac{1}{\sqrt{n!}}\,B_\theta^n\varphi_0^{(\theta)}(x)$, together with the following identity for the Hermite polynomials: $H_n(y)=2yH_{n-1}(y)-H'_{n-1}(y)$, $n\geq1$. Since it does not differ much from other similar derivations, it will not be given here.

It is clear that, for $\theta\in I$, $\varphi_n^{(\theta)}(x), \psi_n^{(\theta)}(x)\in\ltwo$, for all $n\geq0$. Also, these functions belong to the domain of $A_\theta$, $B_\theta$ and of their adjoints, and, as for all pseudo-bosonic operators, we have, as in (\ref{C3}),
\be
\left\{
\begin{array}{ll}
	B_\theta\,\varphi_n^{(\theta)}(x)=\sqrt{n+1}\,\varphi_{n+1}^{(\theta)}(x), \hspace{3.5cm} n\geq 0,\\
	A_\theta\,\varphi_0^{(\theta)}(x)=0,\quad A_\theta\,\varphi_n^{(\theta)}(x)=\sqrt{n}\,\varphi_{n-1}^{(\theta)}(x), \qquad\,\,\,\, n\geq 1,\\
	A_\theta^\dagger\,\psi_n^{(\theta)}(x)=\sqrt{n+1}\,\psi_{n+1}^{(\theta)}(x), \hspace{3.5cm} n\geq 0,\\
	B_\theta^\dagger\,\psi_0^{(\theta)}(x)=0,\quad B_\theta^\dagger\,\psi_n^{(\theta)}(x)=\sqrt{n}\,\psi_{n-1}^{(\theta)}(x), \qquad\,\,\,\, n\geq 1,\\
	N^{(\theta)}\varphi_n^{(\theta)}(x)=n\,\varphi_n^{(\theta)}(x), \hspace{4.6cm} n\geq 0,\\ {N^{(\theta)}}^\dagger\psi_n^{(\theta)}(x)=n\,\psi_n^{(\theta)}(x), \hspace{4.5cm} n\geq 0,
\end{array}
\right.
\label{211}\en
where $N^{(\theta)}=B_\theta A_\theta$ and ${N^{(\theta)}}^\dagger$ is its adjoint. Then, using (\ref{26}) and (\ref{27}), we conclude that
\be
H_\theta\varphi_n^{(\theta)}(x)=E_n^{(\theta)}\varphi_n^{(\theta)}(x), \qquad H_\theta^\dagger\psi_n^{(\theta)}(x)=E_n^{(-\theta)}\psi_n^{(\theta)}(x),
\label{212}\en
where $E_n^{(\theta)}=\omega e^{i\theta}\left(n+\frac{1}{2}\right)$. Notice that $E_n^{(-\theta)}=\overline{E_n^{(\theta)}}$. Hence the eigenvalues of $H_\theta$ and $H_\theta^\dagger$ have, for generic $\theta\in I$, a non zero real and a non zero imaginary part. 

\vspace{2mm}

{\bf Remark:--} If $\theta=0$ everything collapses to the usual quantum harmonic oscillator, as it is clear from (\ref{21}). In this case, if we take $N^{(0)}=\left(\frac{\Omega}{\pi}\right)^{1/4}$, 
\be
\varphi_n^{(0)}(x)=\psi_n^{(0)}(x)=e_n(x)=\frac{1}{\sqrt{2^n\,n!}}\left(\frac{\Omega}{\pi}\right)^{1/4}\,H_n\left(\sqrt{\Omega}\,x\right)e^{-\frac{1}{2}\,\Omega\,x^2},
\label{213}\en
which is the well known $n$-th eigenstate of the quantum harmonic oscillator, as expected.

\vspace{2mm}

Another, also expected, feature of the families $\F_\varphi^{(\theta)}$ and $\F_\psi^{(\theta)}$ is that, with a proper choice of normalization, their vectors are mutually biorthonormal. Indeed if we fix
\be
N^{(\theta)}=\left(\frac{\Omega}{\pi}\right)^{1/4}e^{i\theta/4},
\label{214}\en
we can check that
\be
\langle\varphi_n^{(\theta)},\psi_m^{(\theta)}\rangle=\delta_{n,m},
\label{215}\en
for all $n,m\geq0$ and for all $\theta\in I$. Incidentally we observe that (\ref{214}) gives back the right normalization when $\theta=0$. Formula (\ref{215}) is a consequence of our algebraic settings, \cite{baginbagbook}, but can also be deduced explicitly, by using some properties of the contour integration in the complex plane, see Appendix A.

It is interesting to observe that the functions $\varphi_n^{(\theta)}(x)$ and $\psi_n^{(\theta)}(x)$ are essentially the rotated versions of the eigenstates $e_n(x)$ in (\ref{213}):
\be
\varphi_n^{(\theta)}(x)=e^{i\theta/4}e_n(e^{i\theta/2}x), \qquad \psi_n^{(\theta)}(x)=e^{-i\theta/4}e_n(e^{-i\theta/2}x),
\label{216}\en
for all $n\geq0$. This is in agreement with (\ref{215}):
$$
\langle\varphi_n^{(\theta)},\psi_m^{(\theta)}\rangle=\int_{\mathbb{R}}\overline{\varphi_n^{(\theta)}(x)}\,\psi_m^{(\theta)}(x)dx=\int_{\Gamma_\theta}e_n(z)e_m(z)dz=\int_{\mathbb{R}}e_n(x)e_m(x)dx=\langle e_n,e_m\rangle=\delta_{n,m},
$$
where we have used the results in Appendix A, to which we refer, and the reality of the functions $e_n(x)$. 

The next step in our analysis is to check if the sets  $\F_\varphi^{(\theta)}$ and $\F_\psi^{(\theta)}$ are both complete in $\ltwo$ and if they produce some resolution of the identity. Indeed, the answer to the first question is affirmative: $\F_\varphi^{(\theta)}$ and $\F_\psi^{(\theta)}$ are both complete in $\ltwo$. This follows from a standard argument adopted in several papers, see \cite{baginbagbook} for instance, and originally proposed, in our knowledge, in \cite{kolfom}: if $\rho(x)$ is a Lebesgue-measurable function which is different from zero almost everywhere (a.e.) in $\Bbb R$, and if there exist two positive constants $\delta, C$ such that $|\rho(x)|\leq C\,e^{-\delta|x|}$ a.e. in $\Bbb R$, then the set $\left\{x^n\,\rho(x)\right\}$ is complete in $\ltwo$. We refer to \cite{baginbagbook} for some physical applications of this result. Because of these completeness, the sets $\Lc_\varphi^{(\theta)}=l.s.\{\varphi_n^{(\theta)}(x)\}$ and $\Lc_\psi^{(\theta)}=l.s.\{\psi_n^{(\theta)}(x)\}$, i.e. the linear spans of the functions in $\F_\varphi^{(\theta)}$ and in $\F_\psi^{(\theta)}$, are both dense in $\ltwo$. Now, (\ref{215}) implies that
\be
\sum_{n=0}^\infty\langle f,\varphi_n^{(\theta)}\rangle\langle\psi_n^{(\theta)},g\rangle=\langle f,g\rangle,
\label{217}\en
$\forall f(x)\in\Lc_\psi^{(\theta)}$ and $\forall g(x)\in\Lc_\varphi^{(\theta)}$, which is a sort of weak resolution of the identity, or some extended version of the Parseval identity.

\subsection{A similarity operator}\label{sectII1}

What we have done so far can be restated in terms of a similarity operator, $V_\theta$, unbounded with unbounded inverse. We put
\be
V_\theta e_n(x)=e^{i\theta/4}e_n(e^{i\theta/2}x)=\varphi_n^{(\theta)}(x),
\label{218}\en
$\forall n\geq0$. This action can be extended by linearity to $\Lc_e$, the linear span of the $e_n(x)$'s. Hence $V_\theta$ is densely defined. Its inverse can be again defined on $\Lc_e$. In particular we have
\be
V_\theta^{-1} e_n(x)=e^{-i\theta/4}e_n(e^{-i\theta/2}x)=\psi_n^{(\theta)}(x).
\label{219}\en
Of course $V_\theta^{-1}=V_{-\theta}$. These two formulas imply that $\varphi_n^{(\theta)}(x)\in D(V_\theta^{-1})$ and $\psi_n^{(\theta)}(x)\in D(V_\theta)$ and that
\be
V_\theta^{-1}\varphi_n^{(\theta)}(x)=V_\theta\psi_n^{(\theta)}(x)=e_n(x),
\label{220}\en
$\forall n\geq0$. Furthermore, it is possible to check that $\psi_n^{(\theta)}(x)\in D(V_\theta^\dagger)$, $\varphi_n^{(\theta)}(x)\in D(V_{-\theta}^\dagger)$, and that
\be
V_{-\theta}^\dagger\varphi_n^{(\theta)}(x)=V_\theta^\dagger\psi_n^{(\theta)}(x)=e_n(x),
\label{221}\en
$\forall n\geq0$. Indeed we have, for instance
$$
\langle e_n,e_m\rangle=\langle\varphi_n^{(\theta)},\psi_m^{(\theta)}\rangle=\langle V_\theta e_n,\psi_m^{(\theta)}\rangle=\langle  e_n,V_\theta^\dagger\psi_m^{(\theta)}\rangle,
$$
so that $\langle e_n, e_m-V_\theta^\dagger\psi_m^{(\theta)}\rangle=0$ for all $n$, which implies that $e_m-V_\theta^\dagger\psi_m^{(\theta)}=0$, in view of the completeness of the set $\F_e=\{e_n(x)\}$. The other equality in (\ref{221}) can be proved similarly. We will show later that $V_\theta$ and its inverse are unbounded.

\vspace{2mm}

{\bf Remark:--} In \cite{dapro} the operator $T_\theta=e^{i\frac{\theta}{2}(a^2-{a^\dagger}^2)}$ was introduced (formally). It rotates the functions as follows: $T_\theta f(x)=e^{i\theta/2}f(e^{i\theta}x)$ and it is (again, formally) self-adjoint and unbounded. It is clear that $V_\theta=T_{\theta/2}$.

\vspace{2mm}

To move from formal to rigorous statements, it is possible to show that 
\be
\langle V_\theta f,g\rangle=\langle  f,V_\theta g\rangle,
\label{222}\en
$\forall f(x), g(x)\in\Lc_e$ and $\forall \theta\in I$. To check this, it is sufficient to prove that $\langle V_\theta e_n,e_m\rangle=\langle  e_n,V_\theta e_m\rangle$, $\forall n,m\geq0$. We first rewrite
$$
\langle V_\theta \,e_n,e_m\rangle=\langle \varphi_n^{(\theta)},e_m\rangle=\int_{\mathbb{R}}\overline{\varphi_n^{(\theta)}(x)}\,e_m(x)\,dx=e^{i\theta/4}\int_{\Gamma_\theta}e_n(z)\,e_m\left(e^{i\theta/2}z\right)\,dz,
$$
while
$$
\langle  e_n,V_\theta e_m\rangle=\int_{\mathbb{R}}\overline{e_n(x)}\,\varphi_m^{(\theta)}(x)\,dx=e^{i\theta/4}\int_{\mathbb{R}}e_n(x)\,e_m\left(e^{i\theta/2}x\right)\,dx.
$$
Hence (\ref{222}) follows if 
$$
\int_{\Gamma_\theta}e_n(z)\,e_m\left(e^{i\theta/2}z\right)\,dz=\int_{\mathbb{R}}e_n(x)\,e_m\left(e^{i\theta/2}x\right)\,dx,
$$
i.e. if we can rotate the path of integration. This is possible and can be checked as in Appendix A, using similar strategies and estimates. This check will not be repeated here,

The equality in (\ref{222}) allows us to prove a resolution of the identity not only on the sets $(\Lc_\psi^{(\theta)},\Lc_\varphi^{(\theta)})$, but also to $(\Lc_\psi^{(\theta)},\Lc_e)$. In other words, we can check that
\be
\sum_{n=0}^\infty\langle f,\varphi_n^{(\theta)}\rangle\langle\psi_n^{(\theta)},g\rangle=\langle f,g\rangle,
\label{223}\en
$\forall f(x)\in\Lc_\psi^{(\theta)}$ and $\forall g(x)\in\Lc_e$. To prove this equality, we observe that for any $f(x)\in\Lc_\psi^{(\theta)}$ we have
$$
f(x)={\sum_n}'c_n\psi_n^{(\theta)}(x)=V_{-\theta}\left({\sum_n}'c_n e_n(x)\right),
$$
where we use the prime in${\sum_n}'$ to stress that this is a finite sum. Then $f(x)\in D(V_\theta)$, and $V_\theta f(x)={\sum_n}'c_n e_n(x)\in\Lc_e$. Moreover, $V_{-\theta}V_\theta f(x)=f(x)$. Now, since $V_\theta f(x), g(x)\in\Lc_e$, using (\ref{222}) we get
$$
\langle f,g\rangle=\langle V_{-\theta}V_\theta f,g\rangle=\langle V_\theta f,V_{-\theta}g\rangle=\sum_k\langle V_\theta f,e_k\rangle\langle e_k,V_{-\theta}g\rangle,
$$
using the Parseval identity for $\F_e$. Now, using (\ref{222}) again, we have $\langle e_k,V_{-\theta}g\rangle = \langle V_{-\theta}e_k,g\rangle= \langle\psi_k^{(\theta)},g\rangle$. To  conclude the proof of (\ref{223}) we need to check that  $\langle V_\theta f,e_k\rangle=\langle f,\varphi_k^{(\theta)}\rangle$. This equality does not follows directly from (\ref{222}), since $f(x)\in\Lc_\psi^{(\theta)}$, so that $f(x)\notin\Lc_e$, in general, but can be easily checked directly, using (\ref{215}):
$$
\langle V_\theta f,e_k\rangle={\sum_n}'\overline{c_n}\langle e_n,e_k\rangle={\sum_n}'\overline{c_n}\delta_{k,n},
$$
which can be different from zero if $k$ appears in ${\sum_n}'$. Otherwise is zero. On the other hand
$$
\langle f,\varphi_k^{(\theta)}\rangle={\sum_n}'\overline{c_n}\langle \psi_n^{(\theta)},\varphi_k^{(\theta)}\rangle={\sum_n}'\overline{c_n}\delta_{k,n},
$$
which produces the same result. 

We conclude this list of resolutions of the identity adding to (\ref{217}) and (\ref{223}) the following one
\be
\sum_{n=0}^\infty\langle f,\psi_n^{(\theta)}\rangle\langle\varphi_n^{(\theta)},g\rangle=\langle f,g\rangle,
\label{224}\en
$\forall f(x)\in\Lc_\varphi^{(\theta)}$ and $\forall g(x)\in\Lc_e$, which can be proved with similar a strategy.

\subsection{Bi-coherent states in $\ltwo$}

In this section we will briefly analyse bi-coherent states associated to our lowering operators $A_\theta$ and $B_\theta^\dagger$ in an Hilbert space context. We will show that this is possible here but not, as we will discuss later, for the IQHO. 

In what follows we will consider two alternative approaches. The first is based on a general theorem which has been proven in \cite{bagproc} and then applied in many situations, \cite{bagspringer}. The second produces our states as solutions of some differential equations.

Let us first recall the theorem in \cite{bagproc}, adapted to the present situation.

\begin{thm}\label{theo1}
	Assume that four strictly positive constants $K_\varphi$, $K_\psi$, $r_\varphi$ and $r_\psi$ exist, together with two strictly positive sequences $M_n(\varphi)$ and $M_n(\psi)$, for which
	\be
	\lim_{n\rightarrow\infty}\frac{M_n(\varphi)}{M_{n+1}(\varphi)}=M(\varphi), \qquad \lim_{n\rightarrow\infty}\frac{M_n(\psi)}{M_{n+1}(\psi)}=M(\psi),
	\label{t21}\en
	where $M(\varphi)$ and $M(\psi)$ could be infinity, and such that, for all $n\geq0$,
	\be
	\|\varphi_n^{(\theta)}\|\leq K_\varphi\,r_\varphi^n M_n(\varphi), \qquad \|\psi_n^{(\theta)}\|\leq K_\psi\,r_\psi^n M_n(\psi).
	\label{t22}\en
	Then the series
	\be
	\varphi^{(\theta)}(z;x)=e^{-\frac{1}{2}|z|^2}\sum_{k=0}^\infty\frac{z^k}{\sqrt{k!}}\varphi_k^{(\theta)}(x),\qquad \psi^{(\theta)}(z;x)=e^{-\frac{1}{2}|z|^2}\sum_{k=0}^\infty\frac{z^k}{\sqrt{k!}}\psi_k^{(\theta)}(x),
	\label{t24}\en
	are convergent in $\mathbb{C}$. Moreover, for all $z\in\mathbb{C}$,
	\be
	A_\theta\varphi^{(\theta)}(z;x)=z\varphi^{(\theta)}(z;x), \qquad B_\theta^\dagger \psi^{(\theta)}(z;x)=z\psi^{(\theta)}(z;x).
	\label{t25}\en
	Finally we have
	\be
	\frac{1}{\pi}\int_{\mathbb{C}}\left<f,\varphi^{(\theta)}\right>\left<\psi^{(\theta)},g\right>\,dz=
	\left<f,g\right>,
	\label{t27}\en
	for all $f(x)\in\Lc_\psi^{(\theta)}$ and $g(x)\in\Lc_\varphi^{(\theta)}\cup\Lc_e$.
	
\end{thm}

We only need to prove here that some bounds like those in (\ref{t22}) are satisfied by our functions. Of course, these bounds can only make sense for square-integrable functions. Then they can be satisfied only if we restrict to $\theta\in I$, as we are doing here. Leaving $I$ changes drastically the situation. This case will be discussed in Section \ref{sect3}. As an important side result of our estimates we will be able soon to conclude that the operators $V_\theta$ and $V_\theta^{-1}$ are unbounded, as already stated in Section \ref{sectII1}.

Let us compute now $\|\varphi_n^{(\theta)}\|$, assuming that $\theta\in I$ but $\theta\neq0$. Indeed, the case $\theta=0$ is trivial, since $\|\varphi_n^{(0)}\|=\|e_n\|=1$. Otherwise we have
$$
	\|\varphi_n^{(\theta)}\|^2=\int_{\mathbb{R}}|\varphi_n^{(\theta)}(x)|^2\,dx=\frac{1}{2^n\,n!}\left(\frac{\Omega}{\pi}\right)^{1/2}\int_{\mathbb{R}}H_n\left(e^{-i\theta/2}\sqrt{\Omega}\,x\right) H_n\left(e^{i\theta/2}\sqrt{\Omega}\,x\right)\,e^{-\Omega \cos(\theta)\,x^2}\,dx,
$$
which returns, using the integral 2.20.16.2, page 502, in \cite{prud}
\be
\|\varphi_n^{(\theta)}\|^2=\frac{1}{\sqrt{\cos(\theta)}}\,P_n\left(\frac{1}{\cos(\theta)}\right).
\label{225}\en
Here $P_n(x)$ is the $n$-th Legendre polynomial. Now, using \cite{szego}, Theorem 8.21.1, we know that the asymptotic behaviour (in $n$) of $P_n(x)$, for $x\notin[-1,1]$, is
\be
P_n(x)\simeq(2\pi n)^{-1/2}(x^2-1)^{-1/4}\left(x+\sqrt{x^2-1)}\right)^{n+1/2}.
\label{asbeh}\en
Putting all together we find that, $\forall \theta\in I\setminus\{0\}$,
\be
\|\varphi_n^{(\theta)}\|^2\leq k^{(\theta)}\frac{1}{\sqrt{n}}\left(\frac{2}{\cos(\theta)}\right)^n, \qquad \|\psi_n^{(\theta)}\|^2\leq k^{(-\theta)}\frac{1}{\sqrt{n}}\left(\frac{2}{\cos(\theta)}\right)^n
\label{226}\en
where $k^{(\theta)}$ is a not-so-relevant positive quantity, depending on $\theta$, and the second inequality follows from the identity $\psi_n^{(\theta)}(x)=\varphi_n^{(-\theta)}(x)$. Now, identifying the quantities in (\ref{t22}) is straightforward. We have, for instance $K_\varphi=\sqrt{k^{(\theta)}}$, $r_\varphi=\sqrt{\frac{2}{\cos(\theta)}}$ and $M_n(\varphi)=n^{-1/4}$, so that $M(\varphi)=1$.

We conclude that the series in (\ref{t24}) converge uniformly  for all $z\in\mathbb{C}$, $x\in\mathbb{R}$ and $\theta\in I$ and define square-integrable functions which are eigenstates of $A_\theta$ and $B_\theta^\dagger$. They also resolve the identity on suitable sets, dense in $\ltwo$.

\vspace{2mm}

{\bf Remark:--}
What we have deduced here show that $V_\theta$ and $V_\theta^{-1}$ are unbounded. This is due to the result in (\ref{225}), and to the asymptotic behaviour of $P_n(x)$, (\ref{asbeh}). This last formula shows that, for all $x>1$ (which is our case), $P_n(x)$ diverges with $n$. Hence $V_\theta$ must necessarily be unbounded. Indeed, suppose this is not so and let us call $M_\theta$ the norm of $V_\theta$: $\|V_\theta\|=M_\theta<\infty$. Hence we should have, using (\ref{218}), $\|\varphi_n^{(\theta)}\|=\|V_\theta e_n\|\leq \|V_\theta\|\|e_n\|=M_\theta$, which is finite independently of $n$, which is impossible, as we have seen. Hence $V_\theta$ is unbounded, for all $\theta\in I\setminus\{0\}$. Of course, since $V_\theta^{-1}=V_{-\theta}$, we conclude that $V_\theta^{-1}$ is unbounded as well.

\vspace{2mm}

For our specific annihilation operators it is easy to find a simple explicit form for $\varphi^{(\theta)}(z;x)$ and $\psi^{(\theta)}(z;x)$. We are using explicitly the variable $x$ in these states to stress the fact that we are working in the coordinate space. Solving the differential equation for $A_\theta\varphi^{(\theta)}(z;x)=z\varphi^{(\theta)}(z;x)$, using the differential expression for $A_\theta$, we find
$$
\varphi^{(\theta)}(z;x)=p^{(\theta)}\exp\left\{\sqrt{2\Omega}\,e^{i\theta/2}\,z\,x-\frac{1}{2}e^{i\theta}\Omega \,x^2\right\},
$$
with $p^{(\theta)}$ a suitable normalization, still to be fixed. Moreover, since $\psi^{(\theta)}(z;x)$ is the solution of $B_\theta^\dagger \psi^{(\theta)}(z;x)=z\psi^{(\theta)}(z;x)$, and since $B_\theta^\dagger=A_{-\theta}$, we have
$$
\psi^{(\theta)}(z;x)=\varphi^{(-\theta)}(z;x)=p^{(-\theta)}\exp\left\{\sqrt{2\Omega}\,e^{-i\theta/2}\,z\,x-\frac{1}{2}e^{-i\theta}\Omega \,x^2\right\}.
$$
We can further fix $p^{(\theta)}$ by requiring that $\left<\varphi^{(\theta)},\psi^{(\theta)}\right>=1$. This condition produces a simple gaussian integral, and a possible choice is $p^{(\theta)}=\left(\frac{\Omega}{\pi}\right)^{1/4}e^{i\theta/4}e^{-z_r^2}$, where $z_r=\Re(z)$, the real part of $z$. Summarizing we have
\be
\varphi^{(\theta)}(z;x)=\left(\frac{\Omega}{\pi}\right)^{1/4}\exp\left\{i\,\frac{\theta}{4}-z_r^2+\sqrt{2\Omega}\,e^{i\theta/2}\,z\,x-\frac{1}{2}e^{i\theta}\Omega \,x^2\right\}
\label{227}\en
and
\be
\psi^{(\theta)}(z;x)=\left(\frac{\Omega}{\pi}\right)^{1/4}\exp\left\{-i\,\frac{\theta}{4}-z_r^2+\sqrt{2\Omega}\,e^{-i\theta/2}\,z\,x-\frac{1}{2}e^{-i\theta}\Omega \,x^2\right\},
\label{228}\en
which can be seen as the explicit analytic expressions for the vectors in (\ref{t24}). 
It is evident that, for the allowed values of $\theta$, these functions are both square integrable. It is also clear that, sending $\theta$ to zero, we get
\be
\lim_{\theta,0}\varphi^{(\theta)}(z;x)=\lim_{\theta,0}\psi^{(\theta)}(z;x)=\left(\frac{\Omega}{\pi}\right)^{1/4}\exp\left\{-z_r^2+\sqrt{2\Omega}\,z\,x-\frac{1}{2}\,\Omega \,x^2\right\}=\Phi(z;x),
\label{229}\en
the standard coherent state for the bosonic annihilation operator, with the usual normalization.

\section{The inverted quantum oscillator}\label{sect3}

The Hamiltonian we want to consider in this section is the following:
\be
H=\frac{1}{2}\left(p^2-\Omega^2x^2\right),
\label{31}\en
where, as in (\ref{21}), $\Omega>0$. We have already seen that $H$ can be formally deduced by $H_\theta$ fixing $\theta$ either to $\frac{\pi}{2}$ or to $-\,\frac{\pi}{2}$. For this reason it is natural to define
\be
\varphi_n^{(\pm)}(x)=\varphi_n^{\left(\pm\frac{\pi}{2}\right)}(x)=\frac{e^{\pm i\pi/8}}{\sqrt{2^n\,n!}}\left(\frac{\Omega}{\pi}\right)^{1/4}H_n\left(e^{\pm i\pi/4}\sqrt{\Omega}\,x\right)e^{\mp\,\frac{i}{2}\,\Omega \,x^2}
\label{32}\en
and
\be
\psi_n^{(\pm)}(x)=\psi_n^{\left(\pm\frac{\pi}{2}\right)}(x)=\varphi_n^{(\mp)}(x)=\frac{e^{\mp i\pi/8}}{\sqrt{2^n\,n!}}\left(\frac{\Omega}{\pi}\right)^{1/4}H_n\left(e^{\mp i\pi/4}\sqrt{\Omega}\,x\right)e^{\pm\,\frac{i}{2}\,\Omega \,x^2}.
\label{33}\en
It is clear that
$$
\|\varphi_n^{(\pm)}\|=\|\psi_n^{(\pm)}\|=\infty,
$$
so that none of these functions is square-integrable, of course. However, even if they are not in $\ltwo$, they are connected to the operators $A_\pm$, $B_\pm$ and their adjoints, where
\be
A_\pm=A_{\pm\,\frac{\pi}{2}}=\frac{1}{\sqrt{2\Omega}}\left(e^{\pm i\pi/4}\Omega\,x+i\,e^{\mp i\pi/4}p\right), \,\,\, B_\pm=B_{\pm\,\frac{\pi}{2}}=\frac{1}{\sqrt{2\Omega}}\left(e^{\pm i\pi/4}\Omega\,x-i\,e^{\mp i\pi/4}p\right),
\label{34}\en
and
\be
B_\pm^\dagger=A_\mp, \qquad A_\pm^\dagger=B_\mp.
\label{35}\en
These operators can all be written in terms of the ordinary bosonic operators $c$ and $c^\dagger$ introduced in Section \ref{sect2} as follows:
\be
A_\pm=\frac{c\pm i c^\dagger}{\sqrt{2}}, \qquad B_\pm=\frac{c^\dagger\pm i c}{\sqrt{2}},
\label{36}\en
with $A_\pm^\dagger$ and $B_\pm^\dagger$ deduced as in (\ref{35}). All these operators leave $\scr$ stable. Then we have
\be
[A_\pm,B_\pm]f(x)=f(x),
\label{37}\en
for all $f(x)\in\scr$. However, it is clear that these operators can also be applied to functions which are outside $\scr$, and even outside $\ltwo$. In fact, these operators can also act on 
$\varphi_n^{(\pm)}(x)$ and $\psi_n^{(\pm)}(x)$ and satisfy similar ladder equations as those given in (\ref{211}):
\be
\left\{
\begin{array}{ll}
	A_\pm\,\varphi_0^{(\pm)}(x)=0,\qquad A_\pm\,\varphi_n^{(\pm)}(x)=\sqrt{n}\,\varphi_{n-1}^{(\pm)}(x), \hspace{2cm} n\geq 1,\\
	B_\pm\,\varphi_n^{(\pm)}(x)=\sqrt{n+1}\,\varphi_{n+1}^{(\pm)}(x), \hspace{5cm} n\geq 0,
\end{array}
\right.
\label{38}\en
and
\be
\left\{
\begin{array}{ll}
	B_\pm^\dagger\,\psi_0^{(\pm)}(x)=0,\qquad B_\pm^\dagger\,\psi_n^{(\pm)}(x)=\sqrt{n}\,\psi_{n-1}^{(\pm)}(x), \hspace{2cm} n\geq 1,\\
	A_\pm^\dagger\,\psi_n^{(\pm)}(x)=\sqrt{n+1}\,\psi_{n+1}^{(\pm)}(x), \hspace{5cm} n\geq 0.
\end{array}
\right.
\label{39}\en

Some easy computations show that $H$ in (\ref{31}) can be written in terms of these ladder operators. To simplify the notation we give the results in an operatorial form\footnote{All the operators we are considering in this section can be applied to functions of $\scr$, but not necessarily: they can also act on $\varphi_n^{(\pm)}(x)$ and $\psi_n^{(\pm)}(x)$, and to their linear combinations. The main difference, in this case, is that we are loosing interest to the standard Hilbert space settings usually adopted in quantum mechanics.}. Specializing $H_\theta$ in (\ref{21}) by taking $\theta=\pm\frac{\pi}{2}$ we put
\be
H_\pm=\pm i\Omega\left(B_\pm A_\pm+\frac{1}{2}\,\1\right).
\label{310}\en
We have, as expected
\be
H=H_+=H_-.
\label{311}\en
Just as a check we observe that $$H_+^\dagger=\left(i\Omega\left(B_+ A_++\frac{1}{2}\,\1\right)\right)^\dagger=-i\Omega\left(A_+^\dagger B_+^\dagger+\frac{1}{2}\,\1\right)=-i\Omega\left(B_-A_-+\frac{1}{2}\,\1\right)=H_-,$$
using (\ref{35}). Of course, since $H_-=H_+$, we conclude that $H_+=H_+^\dagger$, at least formally. Of course we have
\be
H_\pm \varphi_n^{(\pm)}(x)=\pm i\Omega \left(n+\frac{1}{2}\right)\varphi_n^{(\pm)}(x),
\label{312}\en
$\forall n\geq0$. Hence the eigenvalues of the IQHO are purely imaginary with both a positive and a negative imaginary part. Of course the functions $\psi_n^{(\pm)}(x)$, which are usually the eigenstates of the adjoint of the {\em original} Hamiltonian, see (\ref{212}), are not so relevant here since the adjoint of $H_+$ is  $H_+$ itself. This is in agreement with the fact that, see (\ref{33}), $\psi_n^{(\pm)}(x)=\varphi_n^{(\mp)}(x)$.

To put the eigenfunctions of $H$ in a more interesting mathematical settings we start defining the following quantities:
\be
\Phi_n^{(\pm)}[f]=\langle\varphi_n^{(\pm)},f\rangle, \qquad \Psi_n^{(\pm)}[g]=\langle\psi_n^{(\pm)},g\rangle,
\label{313}\en
$\forall f(x), g(x)\in\scr$ and $\forall n\geq0$. Here $\langle.,.\rangle$ is the form with extend the ordinary scalar product to {\em compatible pairs}, i.e. to pairs of functions with are, when multiplied together, integrable, but separately they are not (or, at least, one is not). Compatible pairs have been considered in several contributions in the literature. We refer to \cite{pip,anttra} for their appearance in {\em partial inner product spaces}, and to \cite{bagspringer} for some consideration closer (in spirit) to what we are doing here.

The next steps consist in showing that $\Phi_n^{(\pm)}[f]$ and  $\Psi_n^{(\pm)}[g]$ are well defined, linear, and continuous in the natural topology $\tau_\Sc$ in $\scr$. In few words, they are tempered distributions,  $\Phi_n^{(\pm)}, \Psi_n^{(\pm)}\in\Sc'(\mathbb{R})$. In the following we will concentrate on $\Phi_n^{(+)}$ since $\Phi_n^{(-)}$ can be treated in the same way, and since the properties of $\Psi_n^{(\pm)}$ follow from those of $\Phi_n^{(\pm)}$ because of the identity $\psi_n^{(\pm)}(x)=\varphi_n^{(\mp)}(x)$ and of the definition (\ref{313}).

To check that $\Phi_n^{(+)}[f]$ is well defined, we observe that
\be
\left|\Phi_n^{(+)}[f]\right|\leq \frac{(\Omega/\pi)^{1/4}}{\sqrt{2^n\,n!}}\int_{\mathbb{R}}\left|H_n\left(e^{ i\pi/4}\sqrt{\Omega}\,x\right) f(x)\right|\,dx\leq M_n \sup_{x\in\mathbb{R}}(1+|x|)^{n+2}|f(x)|.
\label{314}\en
Here we have defined 
$$
M_n=\frac{(\Omega/\pi)^{1/4}}{\sqrt{2^n\,n!}}\int_{\mathbb{R}}\frac{|H_n(e^{ i\pi/4}\sqrt{\Omega}\,x)|}{(1+|x|)^{n+2}}\,dx.
$$
As we see, in this computation we have multiplied and divided the original integrand function $\left|H_n\left(e^{ i\pi/4}\sqrt{\Omega}\,x\right) f(x)\right|$ for $(1+|x|)^{n+2}$. In this way, since the ratio $\frac{|H_n(e^{ i\pi/4}\sqrt{\Omega}\,x)|}{(1+|x|)^{n+2}}$ has no singularity and decreases to zero for $|x|$ divergent as $|x|^{-2}$, we can conclude that $M_n$ is finite (and positive). 
Before going back to (\ref{314}), we also observe that
$$
\sup_{x\in\mathbb{R}}(1+|x|)^{n+2}|f(x)|=\sup_{x\in\mathbb{R}}\sum_{k=0}^{n+2} \left(
\begin{array}{c}
	n+2  \\
	k  \\
\end{array}
\right) |x|^k|f(x)|=\sum_{k=0}^{n+2} \left(
\begin{array}{c}
	n+2  \\
	k  \\
\end{array}
\right)p_{k,0}(f),
$$
where $p_{k,0}(.)$ is one of the seminorms defining the topology $\tau_\Sc$, see \cite{reed} for instance: $p_{k,l}(f)=\sup_{x\in\mathbb{R}}|x|^k|f^{(l)}(x)|$, $k,l=0,1,2,\ldots$. Of course, all these seminorms are finite for all $f(x)\in\scr$, and $\scr$ is complete with respect to $\tau_\Sc$.

Summarizing we have
$$
\left|\Phi_n^{(+)}[f]\right|\leq M_n\sum_{k=0}^{n+2} \left(
\begin{array}{c}
	n+2  \\
	k  \\
\end{array}
\right)p_{k,0}(f),
$$
which allows us to conclude that $\Phi_n^{(+)}[f]$ is well defined for all $f(x)\in\scr$, as we had to check. 

The linearity of  $\Phi_n^{(+)}$ is clear: $\Phi_n^{(+)}[\alpha f+\beta g]=\alpha\Phi_n^{(+)}[f]+\beta\Phi_n^{(+)}[g]$, for all $f(x), g(x)\in\scr$ and $\alpha,\beta\in\mathbb{C}$.

We are left with the proof that $\Phi_n^{(+)}$ is continuous. For that we have to consider a sequence of functions $\{f_k(x)\in\scr\}$, $\tau_\Sc$-convergent to $f(x)\in\scr$, and check that $\Phi_n^{(+)}[f_k]\rightarrow\Phi_n^{(+)}[f]$ for $k\rightarrow\infty$ in $\mathbb{C}$, for all fixed $n$. The proof of this fact is based on the lemma below, whose proof (quite probably well known) is not easy to find and it is given in Appendix B for completeness.

\begin{lemma}\label{lemma1}
	Given a sequence of functions $\{f_k(x)\in\scr\}$, $\tau_\Sc$-convergent to $f(x)\in\scr$, it follows that $|x|^l|f_k(x)|$ converges, in the norm $\|.\|$ of $\ltwo$, to $|x|^l|f(x)|$, $\forall l\geq0$.
\end{lemma}

We have
$$
\left|\Phi_n^{(+)}[f_k-f]\right|=\left|\langle\varphi_n^{(+)},f_k-f\rangle\right|=\left|\left\langle\frac{\varphi_n^{(+)}}{(1+|x|)^{n+1}},(1+|x|)^{n+1}(f_k-f)\right\rangle\right|,
$$
with an obvious manipulation.
Now, since $\frac{\varphi_n^{(+)}(x)}{(1+|x|)^{n+1}}$ and $(1+|x|)^{n+1}(f_k(x)-f(x))$ are both in $\ltwo$, for all $n,k$, we can use the Schwarz inequality and we get
$$
\left|\Phi_n^{(+)}[f_k-f]\right|\leq\left\|\frac{\varphi_n^{(+)}}{(1+|x|)^{n+1}}\right\|\left\|(1+|x|)^{n+1}(f_k-f)\right\|\rightarrow0
$$ 
when $k\rightarrow\infty$, for all fixed $n\geq0$ because of Lemma \ref{lemma1}.

In the same way, we can prove that $\Phi_n^{(-)}$ is continuous. Hence, as already observed, $\Psi_n^{(\pm)}$ are continuous, too. Then $\Phi_n^{(\pm)},\Psi_n^{(\pm)}\in\Sc'(\mathbb{R})$.

The role of tempered distributions in the context of the IQHO is further clarified by the following result.

\begin{thm}\label{theorem3}
	For each fixed $n\geq0$ the vector $\varphi_n^{(\pm)}(x)$ is a weak limit of $\varphi_n^{(\theta)}(x)$, for $\theta\rightarrow\pm\frac{\pi}{2}$:
	\be
	\varphi_n^{(\pm)}(x)=w-\lim_{\theta,\pm\frac{\pi}{2}}\varphi_n^{(\theta)}(x).
	\label{315}\en
	Analogously,
	\be
	\psi_n^{(\pm)}(x)=w-\lim_{\theta,\pm\frac{\pi}{2}}\psi_n^{(\theta)}(x).
	\label{316}\en
	\end{thm} 
\begin{proof}
	It is sufficient to prove that $\varphi_n^{(+)}(x)=w-\lim_{\theta,+\frac{\pi}{2}}\varphi_n^{(\theta)}(x)$, i.e. that
	\be
	\langle \varphi_n^{(+)}-\varphi_n^{(\theta)},f\rangle\rightarrow0
	\label{317}\en
	for $\theta\rightarrow\frac{\pi}{2}$ for all fixed $n\geq0$ and for $f(x)\in\scr$. The proof follows a similar strategy as that of Lemma \ref{lemma1}. First of all we observe that, 
	$$
	|\varphi_n^{(+)}(x)-\varphi_n^{(\theta)}(x)|\leq \frac{(\Omega/\pi)^{1/4}}{\sqrt{2^n\,n!}}\left(\left|H_n(e^{i\pi/4}\sqrt{\Omega}\, x)\right|+\left|H_n(e^{i\theta/2}\sqrt{\Omega}\, x)\right|\right)\leq\frac{(\Omega/\pi)^{1/4}}{\sqrt{2^n\,n!}}p_n(x),
	$$
	where $p_n(x)$ is a suitable polynomial in $|x|$ of degree $n$, independent of $\theta$, whose expression is not particularly relevant\footnote{To clarify this aspect of the proof, let us consider, for instance $H_3(x)=8x^3-12x$. Hence $|H_3(x)|\leq 8|x|^3+12|x|$ and, therefore $|H_3(e^{i\theta/2}\sqrt{\Omega}\, x)|\leq 8(\Omega)^{3/2}|x|^3+12\sqrt{\Omega}|x|$, for instance.}.
	This estimate implies that the function
	$$
	\chi_n^{(\theta)}(x)=\frac{\varphi_n^{(+)}(x)-\varphi_n^{(\theta)}(x)}{(1+|x|)^{n+1}}
	$$
	is square integrable for all fixed $n$ and for all $\theta\in I$. Therefore, since $(1+|x|)^{n+1}f(x)\in\ltwo$ as well, due to the fact that $f(x)\in\scr$, we have
	$$
	\left|\langle \varphi_n^{(+)}-\varphi_n^{(\theta)},f\rangle\right|=	\left|\left\langle \frac{\varphi_n^{(+)}-\varphi_n^{(\theta)}}{(1+|x|)^{n+1}},(1+|x|)^{n+1}f\right\rangle\right|\leq \|\chi_n^{(\theta)}\|\,\|(1+|x|)^{n+1}f\|,
	$$
	using the Schwarz inequality. Now, to conclude as in (\ref{317}), it is sufficient to show that 
	$\|\chi_n^{(\theta)}\|\rightarrow0$ when $\theta\rightarrow\frac{\pi}{2}$, i.e. that
	$$
	\lim_{\theta,\frac{\pi}{2}}\int_{\mathbb{R}}|\chi_n^{(\theta)}(x)|^2\,dx=0.
	$$
	This is a consequence of the Lebesgue dominated convergence theorem, since it is clear first that $\lim_{\theta,\frac{\pi}{2}}\chi_n^{(\theta)}(x)=0$ a.e. in $x$ and since $|\chi_n^{(\theta)}(x)|^2$ is bounded by an $\Lc^1(\mathbb{R})$ function, in view of what we have shown before. Indeed we have
	$$
	|\chi_n^{(\theta)}(x)|^2=\frac{|\varphi_n^{(+)}(x)-\varphi_n^{(\theta)}(x)|^2}{(1+|x|)^{2n+2}}\leq \frac{(\Omega/\pi)^{1/2}}{2^n\,n!}\frac{p^2_n(x)}{(1+|x|)^{2n+2}},
	$$
	which goes to zero for $|x|$ divergent as $|x|^{-2}$.

\end{proof}

Summarizing the results proved so far we can write that {\em the eigenstates of the IQHO are not square integrable. They define tempered distributions and can be obtained as weak limits of the eigenstates of the Swanson-like Hamiltonian introduced in (\ref{21}).} Hence, the IQHO provides an interesting example of what we have called {\em weak pseudo-bosons} in some recent works, \cite{bagJPA2020,bagspringer}, i.e., essentially, ladder operators obeying a deformed version of the CCR, see (\ref{37}), which produce an interesting functional structure {\em outside} $\ltwo$.

\subsection{Bi-coherent states for the IQHO}

Since we have, for the IQHO, weak pseudo-bosonic operators, and annihilation operators in particular, we can think to construct coherent states of some kind for them. Of course, Theorem \ref{theo1} is no longer available, since the estimates in (\ref{t22}) make no sense, here. This is because, as we have already stressed, $\|\varphi_n^{(\pm)}\|=\|\psi_n^{(\pm)}\|=\infty$. However, as discussed in connection with  other systems, it is still possible to give a slightly different version of this theorem, requiring some bounds not on the norms of $\varphi_n^{(\pm)}(x)$ and $\psi_n^{(\pm)}(x)$, but on some scalar product related to these functions. For instance, if an inequality like the following, analogous to the one in (\ref{t22}) holds,
$$
|\langle \varphi_n^{(+)},f\rangle|\leq A_{\varphi,f}\,r_{\varphi,f}^n M_n(\varphi,f),
$$
for all $f(x)$ in a sufficiently large set of functions, $\E\subset\ltwo$, we can still check that the series
$$
e^{-\frac{1}{2}|z|^2}\sum_{k=0}^\infty\frac{z^k}{\sqrt{k!}}\langle\varphi_k^{(+)},f\rangle
$$ 
converges for all $z\in\mathbb{C}$ and for all $f(x)\in\E$. A similar estimate could be used for $\varphi_k^{(-)}$, of course, producing similar results. In other words, Theorem \ref{theo1} could be extended to our situation. However, we prefer to adopt a different, and possibly more direct, strategy. In particular, once we have seen that the eigenstates of $H=H_{\pm}$ can be found as weak limit of the eigenstates of $H_\theta$, see Theorem \ref{theorem3}, it is natural to check if the weak limit of the bicoherent states in (\ref{227}) and (\ref{228}) produces, or not, some interesting state. In particular, for instance, does $w-\lim_{\theta,\frac{\pi}{2}}\varphi^{(\theta)}(z;x)$ exist?  and is it an eigenstate of $A_+$ with eigenvalue $z$? Is it a tempered distributions? Does any resolution of the identity holds, in this case?

These are the main questions we want to answer. But first of all, it is convenient to start with an useful remark. We put, in analogy with (\ref{32}) and (\ref{33}),
\be
\varphi^{(\pm)}(z;x)=\varphi^{(\pm\frac{\pi}{2})}(z;x)=\left(\frac{\Omega}{\pi}\right)^{1/4}\exp\left\{\pm i\,\frac{\pi}{8}-z_r^2+\sqrt{2\Omega}\,e^{\pm i\pi/4}\,z\,x\mp \frac{i}{2}\,\Omega \,x^2\right\}
\label{318}\en
and
\be
\psi^{(\pm)}(z;x)=\psi^{(\pm\frac{\pi}{2})}(z;x)=\left(\frac{\Omega}{\pi}\right)^{1/4}\exp\left\{\mp i\,\frac{\pi}{8}-z_r^2+\sqrt{2\Omega}\,e^{\mp i\pi/4}\,z\,x\pm \frac{i}{2}\,\Omega \,x^2\right\},
\label{319}\en
so that, not surprisingly, 
\be
\varphi^{(\pm)}(z;x)=\psi^{(\mp)}(z;x).
\label{320}\en
It is clear from these formulas that $|\varphi^{(\pm)}(z;x)|$ and $|\psi^{(\pm)}(z;x)|$ are not polynomials in $|x|$, contrarily  to $|\varphi_n^{(\pm)}(x)|$ and $|\psi_n^{(\pm)}(x)|$. For this reason, $\varphi^{(\pm)}(z;x)$ and $\psi^{(\pm)}(z;x)$ cannot give rise to tempered distributions. Still, how we will show now, they produce certain continuous linear functionals on a rather large set of functions of {\em generalized test functions}, $\V_\rho$. This will be done, at a certain general level, in the next section. After giving some general definitions and results on $\V_\rho$ and its dual, we will come back to the states in (\ref{318}) and (\ref{319}) and to their features in Section \ref{sectbacktoIQHO}.

\subsection{Generalized test functions: general results}

Let $\rho(x)$ be a strictly positive, L-measurable, function, and let us introduce the set
\be
\V_\rho=\left\{f(x)\in\ltwo:\,\rho(x)f(x)\in\ltwo\right\}.
\label{321}\en
First of all we observe that, if $\rho(x)\in\Lc^\infty(\mathbb{R})$, then $\V_\rho=\ltwo$. Indeed, in this case, if we take any $g(x)\in\ltwo$, then $\rho(x)g(x)\in\ltwo$ as well:
$$
\int_{\mathbb{R}}|\rho(x)g(x)|^2\,dx\leq (\sup_{x\in\mathbb{R}}|\rho(x)|)^2\int_{\mathbb{R}}|g(x)|^2\,dx=\|\rho\|^2_\infty\|g\|^2<\infty.
$$
Hence $g(x)\in\V_\rho$: $\ltwo\subseteq\V_\rho$. Vice-versa, if $g(x)\in\V_\rho$, $g(x)\in\ltwo$ by definition. Hence $\V_\rho\subseteq\ltwo$, and our claim follows.

Second, if $\rho(x)$ is continuous, not necessarily bounded, $\V_\rho$ is dense in $\ltwo$. This is because, for any such $\rho(x)$, we have the following inclusion: $D(\mathbb{R})\subseteq\V_\rho$, where  $D(\mathbb{R})$ is the set of all compactly-supported $C^\infty$ functions, which is dense in $\ltwo$.  Hence $\V_\rho$ is also dense. To check this, we take $h(x)\in D(\mathbb{R})$. Hence $h(x)\in\ltwo$, clearly. But now, since $\rho(x)$ is continuous by assumption, $\rho(x)h(x)$ is continuous with bounded support, so that it has a maximum. This implies that $\|\rho h\|<\infty$.

Already these cases show that $\V_\rho$ is a significantly large set, in many interesting situations. We can endow $\V_\rho$ with a topology, $\tau_\rho$, as follows:

\begin{defn} Given a sequence $\{f_n(x)\in\V_\rho\}$, we say that this is $\tau_\rho$-convergent if (i) $\{f_n(x)\in\V_\rho\}$ is $\|.\|$-Cauchy and if (ii) $\{\rho(x)f_n(x)\in\V_\rho\}$ is $\|.\|$-Cauchy.
\end{defn}

We can prove the following result, suggesting that $\tau_\rho$ is a {\em good topology} for us.

\begin{lemma}\label{lemma5}
	Suppose that $\rho^{-1}(x)\in\Lc^\infty(\mathbb{R})$. Then $\V_\rho$ is closed with respect to $\tau_\rho$.
\end{lemma}
\begin{proof}
	Because of the definition of $\tau_\rho$, given a sequence $\{f_n(x)\in\V_\rho\}$ which is $\tau_\rho$-convergent, there exist two functions in $\ltwo$, $f(x)$ and $g(x)$, such that
	$$
	f(x)=\|.\|-\lim_{n,\infty}f_n(x), \qquad 	g(x)=\|.\|-\lim_{n,\infty}\rho(x)f_n(x).
	$$
	We only need to check that $g(x)=\rho(x)f(x)$. For that we start noticing that
	$$
	\|f_n-\rho^{-1}g\|\leq\|\rho^{-1}\|_\infty\|\rho f_n-g\|\rightarrow0,
	$$
	when $n\rightarrow\infty$, because of our assumption on $\rho^{-1}(x)$ and of the definition of $g(x)$.
	Then we have
	$$
	\|f-\rho^{-1}g\|\leq \|f-f_n\|+\|f_n-\rho^{-1}g\|\rightarrow0,
	$$
	when $n\rightarrow\infty$. Hence $f(x)=\rho^{-1}(x)g(x)$, so that $g(x)=\rho(x)f(x)$.

\end{proof}

\vspace{2mm}

{\bf Remark:--} This Lemma shows also that, if $\{f_n(x)\in\V_\rho\}$ is $\tau_\rho$-convergent to $f(x)$, then $\{f_n(x)\in\V_\rho\}$ is also norm convergent to $f(x)$. However, the opposite implication does not hold, without further assumptions on $\rho(x)$. For instance, if $\rho(x)\in\Lc^\infty(\mathbb{R})$, then using the same estimate we deduced before, $\|\rho g\|\leq\|\rho\|_\infty\|g\|$, $g(x)\in\ltwo$, it is clear that if $\{f_n(x)\in\V_\rho\}$ is norm convergent to $f(x)$, it also converges to $f(x)$ in $\tau_\rho$. But, if $\rho(x)\notin\Lc^\infty(\mathbb{R})$, this implication is not automatic.

\vspace{2mm}

From now on we restrict to real $\rho(x)$ satisfying the assumption of Lemma \ref{lemma5}. This is indeed enough for our application. We can now introduce a sort of {\em dual space} of $\V_\rho$, $\Theta_\rho$, as follows:
\be
\Theta_\rho=\{\Phi(x), \mbox{L-measurable: } \Phi(x)\rho^{-1}(x)\in\ltwo\}.
\label{322}\en
It is possible to show that every element of $\Theta_\rho$ defines a continuous linear functional on $\V_\rho$, that is, an  element of $\V_\rho'$. Stated differently, $\Theta_\rho\subseteq\V_\rho'$.

The proof of this claim is easy. Indeed, since
\be
F_\Phi[f]=\langle\Phi,f\rangle=\langle\Phi\rho^{-1},f\rho\rangle,
\label{323}\en
and since, $\forall\Phi(x)\in\Theta_\rho$ and $f(x)\in\V_\rho$ we have $\Phi(x)\rho^{-1}(x),f(x)\rho(x)\in\ltwo$, $F_\Phi[f]$ is well defined and
$$
|F_\Phi[f]|\leq \|\Phi\rho^{-1}\|\,\|f\rho\|.
$$
Linearity of $F_\Phi$ is clear. As for its continuity, let $\{f_n(x)\in\V_\rho\}$ be a sequence $\tau_\rho$ convergent to $f(x)\in\V_\rho$. To show that $F_\Phi[f_k]\rightarrow F_\Phi[f]$ we observe that
$$
|F_\Phi[f_k-f]|=\left|\langle\Phi\rho^{-1},\rho(f_k-f)\rangle\right|\leq \|\Phi\rho^{-1}\|\|\rho(f_k-f)\|\rightarrow0
$$
for $k\rightarrow\infty$.

\subsection{Back to the IQHO}\label{sectbacktoIQHO}

In the last part of Section \ref{sect3}, we fix $\rho(x)=e^{\Omega \,\frac{x^2}{2}}$. This function is clearly continuous, and therefore $\V_\rho$ is dense in $\ltwo$, and strictly included in $\ltwo$. Moreover, it is real and invertible, and its inverse $\rho^{-1}(x)=e^{-\Omega \frac{x^2}{2}}$ is in $\Lc^\infty(\mathbb{R})$.

With this choice we see that $\varphi^{(\pm)}(z;x), \psi^{(\pm)}(z;x)\in\Theta_\rho$. Hence they belong to $\V_\rho'$.

The resolution of the identity for these states is the following:
\be
\int_{\mathbb{C}}\langle f,\psi^{(\pm)}\rangle\langle\varphi^{(\pm)},g\rangle\,\frac{dz}{\pi}=\int_{\mathbb{C}}\langle f,\varphi^{(\pm)}\rangle\langle\psi^{(\pm)},g\rangle\,\frac{dz}{\pi}=\langle f,g\rangle,
\label{324}\en
$\forall f(x),g(x)\in\V_\rho$. We only check the equality
$$
\int_{\mathbb{C}}\langle f,\varphi^{(+)}\rangle\langle\psi^{(+)},g\rangle\,\frac{dz}{\pi}=\langle f,g\rangle,
$$
since the other can be proved in a similar way. We start noticing that
$$
\varphi^{(+)}(z;x)\rho^{-1}(x)=\left(\frac{\Omega}{\pi}\right)^{1/4}\exp\left\{ i\,\frac{\pi}{8}-z_r^2+\sqrt{2\Omega}\,e^{ i\pi/4}\,z\,x- \frac{1}{\sqrt{2}}\,\Omega e^{i\pi/4}\,x^2\right\},
$$
and
$$
\psi^{(+)}(z;x)\rho^{-1}(x)=\left(\frac{\Omega}{\pi}\right)^{1/4}\exp\left\{ -i\,\frac{\pi}{8}-z_r^2+\sqrt{2\Omega}\,e^{ -i\pi/4}\,z\,x- \frac{1}{\sqrt{2}}\,\Omega e^{-i\pi/4}\,x^2\right\}.
$$
Then we have
$$
\langle f,\varphi^{(+)}\rangle=\langle f\rho,\varphi^{(+)}\rho^{-1}\rangle=\frac{e^{-i\,\frac{\pi}{8}-z_r^2}}{(\pi\Omega)^{1/4}}\int_{\tilde\Gamma}\overline{f_\rho\left(\frac{se^{-i\pi/4}}{\sqrt{\Omega}}\right)}\,e^{\sqrt{2}\,z\,s-\frac{1}{\sqrt{2}}\,\Omega e^{-i\pi/4}\,s^2}\, ds,
$$
where we have introduced the simplifying notation $f_\rho(x)=f(x)\rho(x)$ and where $\tilde\Gamma$ is the line which is obtained from the real axis after the change of variable $s=\sqrt{\Omega}\,e^{i\pi/4}\,x$, i.e. the first bisector of the complex plane. In a similar way, calling $g_\rho(x)=g(x)\rho(x)$, we find also
$$
\langle \psi^{(+)},g\rangle=\frac{e^{-i\,\frac{\pi}{8}-z_r^2}}{(\pi\Omega)^{1/4}}\int_{\tilde\Gamma}\,e^{\sqrt{2}\,\overline z\,t-\frac{1}{\sqrt{2}}\,\Omega e^{-i\pi/4}\,t^2}\,g_\rho\left(\frac{te^{-i\pi/4}}{\sqrt{\Omega}}\right) dt,
$$
where again  $t=\sqrt{\Omega}\,e^{i\pi/4}\,x$. Now, if we put these results in $\int_{\mathbb{C}}\langle f,\varphi^{(+)}\rangle\langle\psi^{(+)},g\rangle\,\frac{dz}{\pi}$ and we first integrate in $dz=dz_r\,dz_i$, $z=z_r+iz_i$, we find
$$
\int_{\mathbb{C}}\langle f,\varphi^{(+)}\rangle\langle\psi^{(+)},g\rangle\,\frac{dz}{\pi}=\frac{e^{-i\,\frac{\pi}{4}}}{\sqrt{\Omega}}\int_{\tilde\Gamma}\overline{f_\rho\left(\frac{se^{-i\pi/4}}{\sqrt{\Omega}}\right)}\,g_\rho\left(\frac{se^{-i\pi/4}}{\sqrt{\Omega}}\right)e^{is^2}\,ds,
$$
which can be further simplified by changing the integration variable as in $w=\frac{se^{-i\pi/4}}{\sqrt{\Omega}}$. In this way $\tilde\Gamma$ is mapped back to the real axis, and, recalling the definitions of $f_\rho(x)$ and $g_\rho(x)$, we find the expected result:
$$
\int_{\mathbb{C}}\langle f,\varphi^{(+)}\rangle\langle\psi^{(+)},g\rangle\,\frac{dz}{\pi}=\int_{\mathbb{R}}\overline{f(w)}\,g(w)\,dw=\langle f,g\rangle.
$$

\vspace{2mm}

{\bf Remark:--} It might be interesting to notice that what we have done for the bicoherent states of the IQHO could be done also for the eigenstates of its Hamiltonian $H$. In other words we could check that the states $\varphi_n^{(\pm)}(x)$ and $\psi_n^{(\pm)}(x)$ in (\ref{32}) and (\ref{33}) are not only tempered distributions, but also elements of $\V_\rho'$. Checking this is much simpler, but maybe less interesting (once we have proved the stronger result!). For this reason we will not give the details of this result here.

\section{Conclusions}\label{sect4}

We have shown how the IQHO can be studied adopting ideas coming from distribution theory, and how the eigenfunctions and the bi-coherent states can be recovered as a weak limit of the analogous objects of a Swanson-like Hamiltonian. Several mathematical properties of these states have been studied, from the point of view of tempered distribution and from that of continuous functional on a certain space of functions, introduced for the IQHO but discussed here in some generality. Our approach produces an alternative way to deal with a system which has attracted very much the attention of many scientists in recent years, putting the IQHO, we believe, on a mathematically firm ground.

\section*{Acknowledgements}

The author acknowledges partial support from Palermo University and from G.N.F.M. of the INdAM.

\renewcommand{\theequation}{A.\arabic{equation}}

\section*{Appendix A: proof of the biorthonormality}\label{appendixA}

This appendix is devoted to the proof of formula (\ref{215}), $\langle\varphi_n^{(\theta)},\psi_m^{(\theta)}\rangle=\delta_{n,m}$. We start rewriting, using (\ref{29}) and (\ref{210}),
$$
\langle\varphi_n^{(\theta)},\psi_m^{(\theta)}\rangle=\frac{\overline{N^{(\theta)}}\,N^{(-\theta)}}{\sqrt{2^{n+m}\,n!\,m!}}\int_{\mathbb{R}}H_n\left(e^{-i\theta/2}\sqrt{\Omega}\,x\right)H_m\left(e^{-i\theta/2}\sqrt{\Omega}\,x\right)e^{-\Omega e^{-i\theta}x^2}\,dx=
$$
$$
=\frac{\overline{N^{(\theta)}}\,N^{(-\theta)}\,e^{i\theta/2}}{\sqrt{2^{n+m}\,n!\,m!\,\Omega}}\int_{\Gamma_\theta}H_n\left(z\right)H_m\left(z\right)e^{-z^2}\,dz,
$$
where we have introduced the complex variable $z=e^{-i\theta/2}\sqrt{\Omega}\,x$ and the straight line $\Gamma_\theta$, which is what should replace $\mathbb{R}$ because of this change of variable: $\Gamma_\theta=\left\{z=(\cos(\theta/2)-i\sin(\theta/2))x, \,x\in\mathbb{R}\right\}$, see the figure. Using some trick in complex integration it is possible to show that we can {\em rotate back} $\Gamma_\theta$, getting  
\be
\int_{\Gamma_\theta}H_n\left(z\right)H_m\left(z\right)e^{-z^2}\,dz=\int_{\mathbb{R}}H_n\left(x\right)H_m\left(x\right)e^{-x^2}\,dx=2^n\,n!\sqrt{\pi}\,\delta_{n,m}.
\label{a1}\en
Hence we get
$$
\langle\varphi_n^{(\theta)},\psi_m^{(\theta)}\rangle=\overline{N^{(\theta)}}\,N^{(-\theta)}\,e^{i\theta/2}\sqrt{\frac{\pi}{\Omega}}\,\delta_{n,m},
$$
which returns (\ref{215}) if $N^{(\theta)}$ is chosen as in (\ref{214}). To prove (\ref{a1}) we start constructing the   curve $\Sigma(R)$ in figure, $R$ finite, in the complex plane.

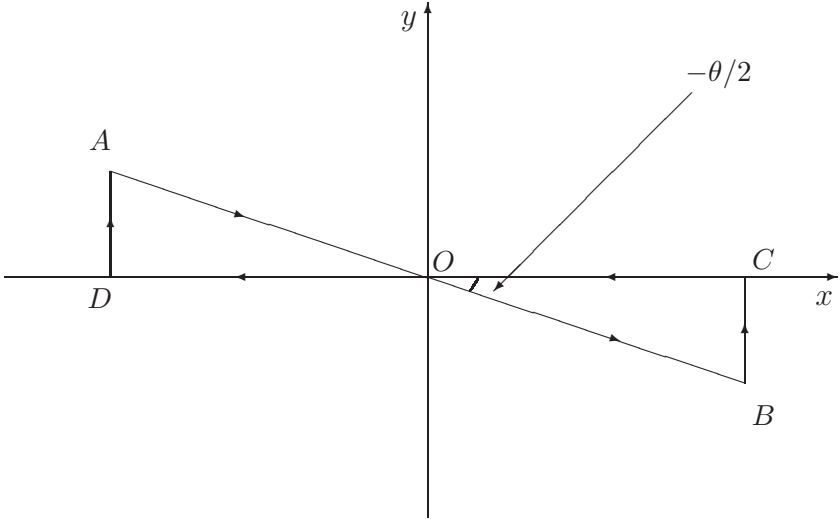
\vspace*{15mm}
\begin{picture}(450,135)
	\put(20,66){\vector(1,0){315}} \put(180,-25){\vector(0,1){195}} \put(173,163){\makebox(0,0){$y$}} \put(330,58){\makebox(0,0){$x$}}

	\put(60,106){\line(3,-1){240}}
	
	\put(300,66){\line(0,-1){40}}
	\put(60,66){\line(0,1){40}}
	
	\put(307,14){\makebox(0,0){\small $B$}} \put(56,118){\makebox(0,0){\small$A$}}
	\put(307,73){\makebox(0,0){\small $C$}} \put(56,58){\makebox(0,0){\small$D$}}
	\put(186,72){\makebox(0,0){\small$O$}}
	
	\put(110,66){\vector(-1,0){3}}\put(250,66){\vector(-1,0){3}}
	\put(108,90){\vector(3,-1){3}}\put(250,43){\vector(3,-1){3}}
	\put(60,86){\vector(0,1){3}}\put(300,46){\vector(0,1){3}}
	
\qbezier[300](199,66)(199,65)(196,61)	

\put(290,143){\makebox(0,0){\small$-\theta/2$}}

\put(280,136){\vector(-1,-1){75}}

\end{picture}

\vspace*{8mm}

These are the coordinates of the points in figure: $A=-B=-R(1,-\tan(\theta/2))$, $D=-C=-R(1,0)$. The real axis can be recovered by the horizontal line $\gamma_{[C,D[}(R)$ in figure for $R$ diverging. Notice that, going from $C$ to $D$, we are moving in the opposite direction with respect to positive direction of the real axis. The limit $R\rightarrow\infty$ of the oblique line $\gamma_{[A,B[}(R)$ returns $\Gamma_\theta$. The vertical segments $\gamma_{[B,C[}(R)$ and $\gamma_{[D,A[}(R)$, with their arrows, close the circuit $\Sigma(R)$ given in the picture.

Now, since the function $H_{n,m}(z)=H_n\left(z\right)H_m\left(z\right)e^{-z^2}$ is analytical and has no singularity inside $\Sigma(R)$, it follows that
$$
0=\int_{\Sigma(R)}H_{n,m}(z)\,dz=$$
$$=\int_{\gamma_{[A,B[}(R)}H_{n,m}(z)\,dz+\int_{\gamma_{[B,C[}(R)}H_{n,m}(z)\,dz+\int_{\gamma_{[C,D[}(R)}H_{n,m}(z)\,dz+\int_{\gamma_{[D,A[}(R)}H_{n,m}(z)\,dz.
$$
The two integrals along the vertical lines are related. Indeed we have, using the parity of the Hermite polynomials, $H_n(-x)=(-1)^nH_n(x)$ and the parametric expressions for $\gamma_{[B,C[}(R)$ and $\gamma_{[D,A[}(R)$, $z=R+iy$ with $y_B\leq y<0$ and $z=-R+iy$ with $0\leq y<y_A=-y_B$ respectively, that
$$
\int_{\gamma_{[D,A[}(R)}H_{n,m}(z)\,dz=(-1)^{n+m}\int_{\gamma_{[B,C[}(R)}H_{n,m}(z)\,dz.
$$
Hence, it is sufficient to check that one of these integral goes to zero for $R\rightarrow\infty$ to conclude that the same is true for the other. In particular, after some minor manipulations, we deduce that
$$
\left|\int_{\gamma_{[B,C[}(R)}H_{n,m}(z)\,dz\right|\leq \frac{1}{e^{R^2}}\int_{0}^{R\tan(|\theta|/2)}|H_n(R-iy)H_m(R-iy)|e^{y^2}dy,
$$
which, for our values of $\theta$, is of the form $\frac{\infty}{\infty}$ for $R$ diverging. If we then use the de l'Hopital theorem, we see that, since $0<\tan(|\theta|/2)<1$, the right-hand side above goes to zero. This means that, taking the limit for $R\rightarrow\infty$ of $\int_{\Sigma(R)}H_{n,m}(z)\,dz$, we have
$$
0=\lim_{R,\infty}\int_{\gamma_{[A,B[}(R)}H_{n,m}(z)\,dz+\lim_{R,\infty}\int_{\gamma_{[C,D[}(R)}H_{n,m}(z)\,dz,
$$
which implies (\ref{a1}).

\renewcommand{\theequation}{B.\arabic{equation}}

\section*{Appendix B: proof of Lemma \ref{lemma1}}\label{appendixB}

To prove the assertion we need to check that
$$
I_k^{(l)}=\int_{\mathbb{R}}\left|x^lf_k(x)-x^lf(x)\right|^2\,dx \rightarrow 0,
$$
when $k\rightarrow\infty$, $\forall l\geq0$.

Indeed, since $1+|x|>|x|$ and $1+|x|\geq1$, we have 
$$
\left|x^lf_k(x)-x^lf(x)\right|\leq  (1+|x|)^l\left|f_k(x)-f(x)\right|\leq \frac{1}{1+|x|}\sup_{x\in\mathbb{R}}(1+|x|)^{l+1}\left|f_k(x)-f(x)\right|\leq
$$
$$
\leq \frac{1}{1+|x|}\sum_{n=0}^{l+1} \left(
\begin{array}{c}
	l+1  \\
	n  \\
\end{array}
\right)p_{n,0}(f_k-f),
$$
where $p_{n,0}(.)$ is, as in Section \ref{sect3}, one of the seminorms defining the topology $\tau_\Sc$. Since, by assumption, $f_k(x)$ is $\tau_\Sc$-convergent to $f(x)$, $p_{n,0}(f_k-f)\rightarrow0$ when $k\rightarrow\infty$, $\forall n\geq0$. Then, calling
$$
D_{l+1}(f_k-f)=\sum_{k=0}^{l+1} \left(
\begin{array}{c}
	l+1  \\
	n  \\
\end{array}
\right)p_{n,0}(f_k-f),
$$
we have that $D_{l+1}(f_k-f)\rightarrow0$, $\forall l\geq0$, when $k\rightarrow\infty$. Hence,
$$
I_k^{(l)}\leq \left[D_{l+1}(f_k-f)\right]^2\int_{\mathbb{R}}\frac{dx}{(1+|x|)^2}=2\left[D_{l+1}(f_k-f)\right]^2\rightarrow0,
$$
as we had to check.

\renewcommand{\theequation}{A.\arabic{equation}}

\section*{Appendix C: definition of $\D$ pseudo-bosons}\label{appendixC}

Let $\Hil$ be a given Hilbert space with scalar product $\left<.,.\right>$ and related norm $\|.\|$. 

\vspace{2mm}

Let $a$ and $b$ be two operators
on $\Hil$, with domains $D(a)$ and $D(b)$ respectively, $a^\dagger$ and $b^\dagger$ their adjoint, and let $\D$ be a dense subspace of $\Hil$
such that $a^\sharp\D\subseteq\D$ and $b^\sharp\D\subseteq\D$, where with $x^\sharp$ we indicate $x$ or $x^\dagger$. Of course, $\D\subseteq D(a^\sharp)$
and $\D\subseteq D(b^\sharp)$.

\begin{defn}\label{defc21}
	The operators $(a,b)$ are $\D$-pseudo bosonic  if, for all $f\in\D$, we have
	\be
	a\,b\,f-b\,a\,f=f.
	\label{C1}\en
\end{defn}

When CCR are replaced by (\ref{C1}), it is necessary to impose some reasonable conditions which are verified in explicit models. In particular, our starting assumptions are the following:

\vspace{2mm}

{\bf Assumption $\D$-pb 1.--}  there exists a non-zero $\varphi_{ 0}\in\D$ such that $a\,\varphi_{ 0}=0$.

\vspace{1mm}

{\bf Assumption $\D$-pb 2.--}  there exists a non-zero $\Psi_{ 0}\in\D$ such that $b^\dagger\,\Psi_{ 0}=0$.

\vspace{2mm}

It is obvious that, since $\D$ is stable under the action of the operators introduced above,  $\varphi_0\in D^\infty(b):=\cap_{k\geq0}D(b^k)$ and  $\Psi_0\in D^\infty(a^\dagger)$, so
that the vectors \be \varphi_n:=\frac{1}{\sqrt{n!}}\,b^n\varphi_0,\qquad \Psi_n:=\frac{1}{\sqrt{n!}}\,{a^\dagger}^n\Psi_0, \label{A2}\en
$n\geq0$, can be defined and they all belong to $\D$. Then, they also belong to the domains of $a^\sharp$, $b^\sharp$ and $N^\sharp$, where $N=ba$. We see that, from a practical point of view, $\D$ is the natural space to work with and, in this sense, it is even more relevant than $\Hil$. Let's put $\F_\Psi=\{\Psi_{ n}, \,n\geq0\}$ and
$\F_\varphi=\{\varphi_{ n}, \,n\geq0\}$.
It is  simple to deduce the following lowering and raising relations:
\be
\left\{
\begin{array}{ll}
	b\,\varphi_n=\sqrt{n+1}\varphi_{n+1}, \qquad\qquad\quad\,\, n\geq 0,\\
	a\,\varphi_0=0,\quad a\varphi_n=\sqrt{n}\,\varphi_{n-1}, \qquad\,\, n\geq 1,\\
	a^\dagger\Psi_n=\sqrt{n+1}\Psi_{n+1}, \qquad\qquad\quad\, n\geq 0,\\
	b^\dagger\Psi_0=0,\quad b^\dagger\Psi_n=\sqrt{n}\,\Psi_{n-1}, \qquad n\geq 1,\\
\end{array}
\right.
\label{C3}\en as well as the eigenvalue equations $N\varphi_n=n\varphi_n$ and  $N^\dagger\Psi_n=n\Psi_n$, $n\geq0$. In particular, as a consequence
of these last two equations,  if we choose the normalization of $\varphi_0$ and $\Psi_0$ in such a way $\left<\varphi_0,\Psi_0\right>=1$, we deduce that
\be \left<\varphi_n,\Psi_m\right>=\delta_{n,m}, \label{A4}\en
for all $n, m\geq0$. Hence $\F_\Psi$ and $\F_\varphi$ are biorthogonal.

The analogy with ordinary bosons suggests us to consider the following:

\vspace{2mm}

{\bf Assumption $\D$-pb 3.--}  $\F_\varphi$ is a basis for $\Hil$.

\vspace{1mm}

This is equivalent to requiring that $\F_\Psi$ is a basis for $\Hil$ as well. However, several  physical models show that $\F_\varphi$ is {\bf not} a basis for $\Hil$, but it is still complete in $\Hil$. This suggests to adopt the following weaker version of  Assumption $\D$-pb 3, \cite{baginbagbook}:

\vspace{2mm}

{\bf Assumption $\D$-pbw 3.--}  For some subspace $\G$ dense in $\Hil$, $\F_\varphi$ and $\F_\Psi$ are $\G$-quasi bases.

\vspace{2mm}
This means that, for all $f$ and $g$ in $\G$,
\be
\left<f,g\right>=\sum_{n\geq0}\left<f,\varphi_n\right>\left<\Psi_n,g\right>=\sum_{n\geq0}\left<f,\Psi_n\right>\left<\varphi_n,g\right>,
\label{A4b}
\en
which can be seen as a weak form of the resolution of the identity, restricted to $\G$. Of course, if $f\in\G$ is orthogonal to all the $\varphi_n$'s, or to all the $\Psi_n$'s, then (\ref{A4b}) implies that $f=0$. Hence $\F_\varphi$ and $\F_\Psi$ are complete in $\G$.

We refer to \cite{bagbookPT,bagspringer} for more details.

\end{document}